\documentclass[11pt,a4dutch]{article}
\usepackage{amsmath,amsthm,amssymb,txfonts,bm}
\usepackage{color}
\usepackage{paralist}
\usepackage{a4wide}
\usepackage{graphicx}
\usepackage{url}

\newlength{\minispacing}
\newcommand{\set}[2]{\{ \hspace{\minispacing} #1 \,  | \,#2 \hspace{\minispacing}\}}
\newcommand{\define}{\mathrel{:=}}
\newcommand{\abs}[1]{\left\vert \hspace{\minispacing} #1 \hspace{\minispacing} \right\vert}
\newcommand{\lref}[1]{(\ref{#1})}
\newcommand{\R}{\mathbb{R}}
\newcommand{\sep}{R}
\newcommand{\sset}[1]{\{ #1 \}}
\newcommand{\norm}[1]{\| #1 \|}
\newcommand{\onenorm}[1]{\norm{#1}_1}
\newcommand{\twonorm}[1]{\norm{#1}_2}

\newcommand{\bx}{\bar{x}}   
\newcommand{\tx}{\hat{x}}
\newcommand{\by}{\bar{y}}
\newcommand{\bmu}{\bar{\mu}}
\newcommand{\bs}{\bar{s}}

\theoremstyle{plain}
\newtheorem{theorem}{Theorem}
\newtheorem{lemma}{Lemma}

\newtheorem{fact}{Fact}

\newtheorem{claim}{Claim}

\newtheoremstyle{myremark}{\topsep}{\topsep}{\normalfont}{0pt}{\sffamily}{. }{ }{}

\theoremstyle{myremark}
\newtheorem{rem}{Remark}

\title{A Still Simpler Way of Introducing Interior-Point Method for Linear
Programming}

\author{Kurt Mehlhorn\thanks{E-mail: mehlhorn@mpi-inf.mpg.de}\\
Max-Planck-Institut f{\"u}r Informatik,\\
Saarbr{\"u}cken,\\
Germany
\and
Sanjeev Saxena\thanks{E-mail: ssax@iitk.ac.in}\\
Computer Science and
Engineering,\\ Indian Institute of Technology,\\
Kanpur, INDIA-208 016}

\date{\today}

\begin{document}
\maketitle
\maketitle{}

\subsection*{\centering{Abstract}}

Linear programming is now included in algorithm undergraduate and
postgraduate courses for computer science majors.  We give a
self-contained treatment of an interior-point method which is
particularly tailored to the typical mathematical background of CS
students. In particular, only limited knowledge of linear algebra and
calculus is assumed. 

\section{Introduction}

Terlaky \cite{terlaky} and Lesaja \cite{lesaja} have suggested simple
ways to teach interior-point methods. In this paper, we suggest an
alternative and maybe 
still simpler way which is particularly tailored to the typical
mathematical background of CS students. In particular, only limited
knowledge of linear algebra and 
calculus is
assumed. We have selected most of the material from popular
textbooks~\cite{roos,van,BT,karloff,ye,saigal} to assemble a
self-contained presentation of an interior point method-- little of this
material is new. 

The canonical \emph{linear programming problem} is to 
\begin{equation} \label{primal problem} 
\text{minimize $c^Tx$ subject to $Ax=b$ and $x\geq 0$.} 
\end{equation} 
Here, $A$ is an $m\times n$ matrix, $c$ and $x$ are $n$-dimensional, and
$b$ is an $m$-dimensional vector. A \emph{feasible solution} is any
vector $x$ with $Ax = b$ and $x \ge 0$. The problem is \emph{feasible} if
there is a feasible solution, and \emph{infeasible} otherwise. A feasible
problem is \emph{unbounded} (or more precisely the corresponding
objective function is unbounded) if for every real $z$, there is a
feasible $x$ with $c^T x \le z$, and \emph{bounded} otherwise. 

In our presentation, we first assume that feasible solutions to the
primal and the corresponding dual LP satisfying a certain set of
properties (properties (I1) to (I3) in Section~\ref{sec:invariants}) are
available. We then show how to iteratively improve these solutions in
Sections~\ref{sec:iteration} and~\ref{sec:invariants}.  In each iteration
the gap between the primal and the dual objective value is reduced by a
factor $1 - O(1/\sqrt{n})$, where $n$ is the number of variables. The
iterative improvement scheme leads to solutions that are arbitrarily
close to optimality. In Sections~\ref{sec:init} and~\ref{sec:extraction} we
discuss how to find the appropriate initial solutions and how to extract
an optimal solution from a sufficiently good solution by rounding. 
Either or both these sections may be skipped in a first course. 

\begin{rem} It is easy to deal with maximization instead of
minimization and with inequality constraints. Indeed, maximize $c^Tx$ is
equivalent to minimize $-c^Tx$. Constraints of type $\alpha_1x_1+{\ldots}
+\alpha_nx_n\leq \beta$ can be replaced by $\alpha_1x_1+{\ldots}
+\alpha_nx_n+\gamma = \beta$ with a new (slack) variable $\gamma\geq 0$.
Similarly, constraints of type $\alpha_1x_1+{\ldots} +\alpha_nx_n\geq
\beta$ can be replaced by $\alpha_1x_1+{\ldots} +\alpha_nx_n-\gamma =
\beta$ with a (surplus) variable $\gamma\geq 0$.\end{rem}

\noindent We consider another problem, the \emph{dual problem}, which is 
\begin{equation}\label{dual problem} \text{
maximize $b^T y$, subject to $A^Ty+s=c$, with variables $s\geq 0$ and
unconstrained variables $y$.} \end{equation}

The vector $y$ has $m$ components and the vector $s$ has $n$ components.
We will call the original problem the \emph{primal problem}.

\begin{claim}[Weak Duality] \label{first}\label{weak duality}
If $x$ is a solution of $Ax=b$ with $x\geq 0$ and $(y,s)$ is a solution
of $A^Ty+s=c$ with $s\geq 0$, then
\begin{compactenum}[1.]
\item $x^Ts=c^Tx-b^Ty$, and 
\item $b^Ty\leq c^Tx$, with equality if and only if $s_ix_i=0$ for all $i$s.
\end{compactenum}
\end{claim}
\begin{proof} We multiply $s=c-A^Ty$ with $x^T$ from the left and obtain 
\[ x^Ts=x^Tc- x^T(A^Ty)=c^Tx-(x^TA^T)y= c^T x - (A x)^T y = c^Tx-b^Ty .\]
As $x,s\geq 0$, we have $x^Ts\geq 0$, and hence, $c^Tx\geq b^Ty$.

Equality will hold if $x^Ts=0$, or equivalently, $\sum_i s_ix_i=0$. Since
$s_i,x_i\geq 0$, $\sum_i s_ix_i=0$ if and only if $s_ix_i=0$ for all $i$.
\end{proof}

If $x$ is a feasible solution of the primal and $(y,s)$ is a feasible
solution of the dual, the difference $c^T x - b^T y$ is called the
\emph{objective value gap} of the solution pair. Thus, if the objective
values of a primal feasible and a dual feasible solution are the same,
then both solutions are optimal. Actually, from the Strong Duality
Theorem, if both primal and dual solutions are optimal, then the equality
will hold. We will prove the Strong Duality Theorem in Section~\ref{part}
(Corollary~\ref{strong:dual}).

If the primal and the dual are both feasible, neither of them can be
unbounded as by Claim~\ref{first}, the objective value of all dual
feasible solutions are less than or equal to the objective values of any
primal feasible solution.  As a consequence: 
If the primal and the dual are feasible, both are bounded. 
If the primal is unbounded, the dual is infeasible, and if the dual is
unbounded, the primal is infeasible. It may happen that both problems are
infeasible. It is also true, that if the primal is feasible and bounded, the dual is feasible and bounded, and vice versa. This is a consequence of strong duality. 

\emph{We will proceed under the assumption that the primal as well as the
dual problem are bounded and feasible.} This allows us to concentrate on
the core of the interior point method, the iterative improvement scheme.
We come back to this point in Section~\ref{sec:init}.

Claim~\ref{first} implies, that if we are able to find a solution to the 
following system of equations and inequalities \[ Ax=b,\ A^Ty+s=c,\ x_i
s_i = 0 \text{ for all $i$},\ x\geq 0,\ s\geq 0, \] we will get optimal
solutions of both the original primal and the dual problem. Notice that
the constraints $x_i s_i = 0$ are nonlinear and hence it is not clear
whether we have made a step towards the solution of our problem. The idea
is now to relax the conditions $x_i s_i = 0$ to the conditions $x_i s_i
\approx \mu$ (with the exact form of this equation derived in the next
section), where $\mu \ge 0$ is a parameter. We obtain \[ (P_\mu) \quad
Ax=b,\ A^Ty+s=c,\ x_i s_i \approx \mu \text{ for all $i$},\ x > 0,\ s >
0.\] We will show: 
\begin{compactenum}[1.]
\item (initial solution) For a suitable $\mu$, it is easy to find a
solution to the problem $P_\mu$. This will be the subject of
Section~\ref{sec:init}.
\item (iterative improvement) Given a solution $(x,y,\mu)$ to $P_\mu$,
one can find a solution $(x',y',s')$ to $P_{\mu'}$, where $\mu'$ is
substantially smaller than $\mu$. This will be the subject of
Sections~\ref{sec:iteration} and~\ref{sec:invariants}. Applying this step
repeatedly, we can make $\mu$ arbitrarily small. 
\item (final rounding) Given a solution $(x,y,\mu)$ to $P_\mu$ for
sufficiently small $\mu$, one can extract an exact solutions for the
primal and the dual problem. This will be the subject of
Section~\ref{part}. 
\end{compactenum}\smallskip

\begin{figure}[t]
\centerline{\includegraphics[width=0.6\textwidth]{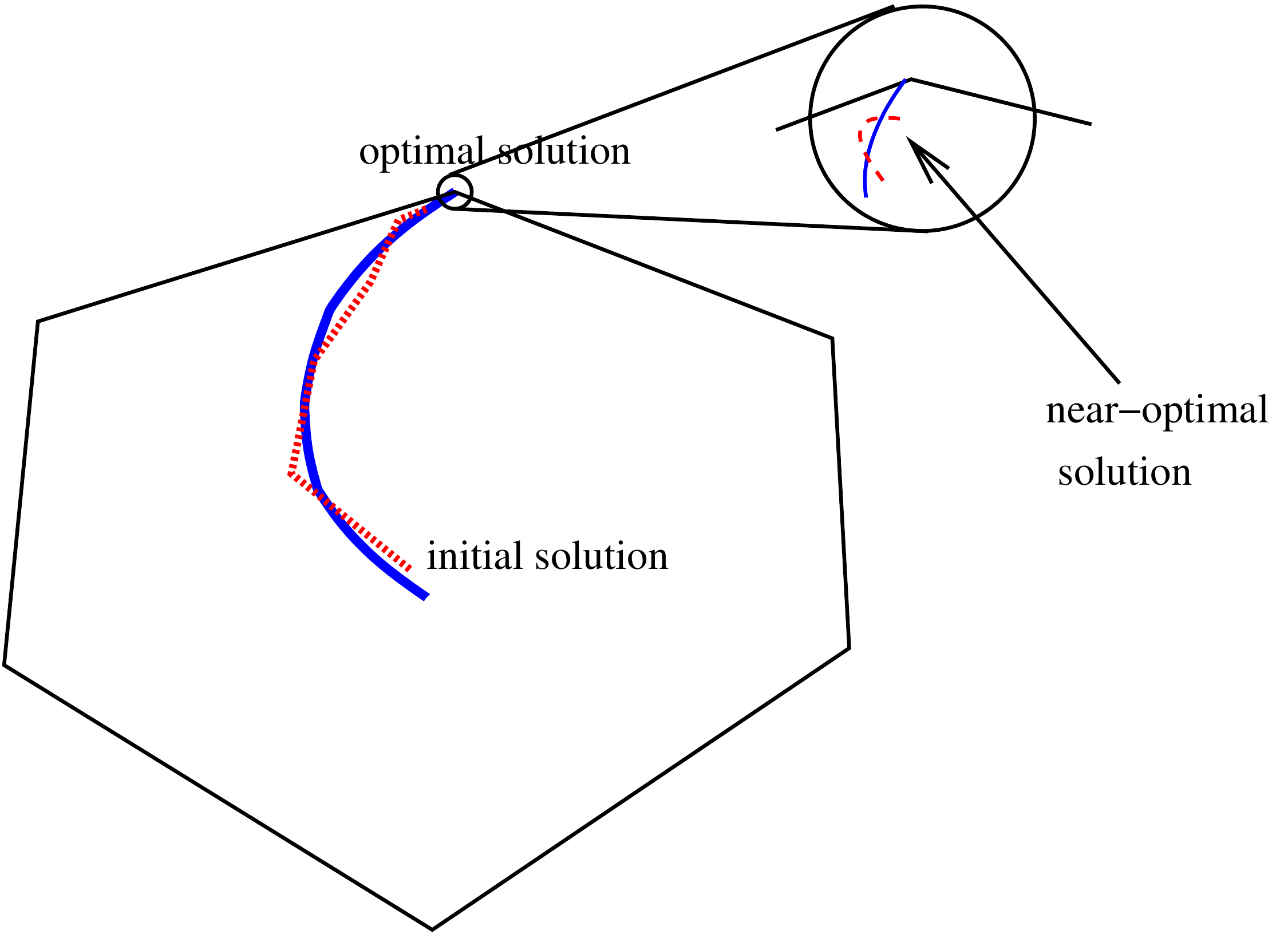}}
\caption{\label{central path} The interior of the polygon comprises all
points $(x,y,s)$ satisfying $Ax = b$ and $A^T y + s = c$, $x > 0$, and $s
> 0$. The blue (bold) line consists of all points in this polygon with
$x_i s_i = \mu$ for all $i$ and some $\mu > 0$. These points trace a line
inside the polygon that ends in an optimal point. The optimal solution
lies on the boundary of the polygon (in the figure, the optimal point is
a vertex of the polygon) and satisfies $x_i s_i = 0$ for all $i$.  The
red (dashed) line illustrates the steps of the algorithm. It follows the
blue (bold) line in discrete steps. The close-up shows the situation near
the optimal solution. The algorithm stops tracing the blue curve and
rounds to the near-optimal red solution obtained at this point of time to
an optimal solution. }
\end{figure}

For the iterative improvement, it is important that $x > 0$ and $s > 0$.
For this reason, we replace the constraints $x \ge 0$ and $s \ge 0$ by $x
> 0$ and $s > 0$ when defining problem $P_\mu$ (see Figure~\ref{central
path}).

Note that $x_i s_i \approx \mu$ for all $i$ implies $b^Ty-c^Tx \approx
n\mu$ by Claim~\ref{first}. Thus, repeated application of iterative
improvement will make the gap between the primal and dual objective
values arbitrarily small. 

\emph{Throughout the paper we assume that the rows of $A$ are linearly
independent and that $n > m$}, i.e., we have more variables than
constraints.\footnote{ Indeed, we can use Gaussian elimination to remove
superfluous constraints and to make the rows of $A$ independent. Assume
first that $A$ contains a row $i$ in which all entries are equal to zero.
If $b_i$ is also zero, we simply delete the row. If $b_i$ is nonzero, the
system of equations has no solution, and we declare the problem
infeasible and stop. Now, every row of $A$ contains a nonzero entry, in
particular, the first row. We may assume that $a_{11}$ is nonzero.
Otherwise, we interchange two columns. We multiply the $i$th equation by
$-\frac{a_{11}}{a_{i1}}$ and subtract the first equation. In this way,
the first entry of all equations but the first becomes zero. If any row
of $A$ becomes equal to the all zero vector, we either delete the
equation or declare the problem infeasible. We now proceed in the same
way with the second equation. We first make sure that $a_{22}$ is nonzero
by interchanging columns if necessary. Then we multiply the $i$th
equation (for $i > 2$) by $-\frac{a_{22}}{a_{21}}$ and subtract the
second equation. And so on. In the end, all remaining equations will be
linearly independent. Equivalently, the resulting matrix will have full
row-rank. 

We now have $m$ constraints in $n$ variables with $n \ge m$. If $n = m$,
the system $Ax =b$ has a unique solution (recalling that $A$ has full
row-rank and is hence invertible). We check whether this solution is
non-negative. If so, we have solved the problem. Otherwise, we declare
the problem infeasible. So, we may from now on assume $n > m$ (more
variables than constraints).}

\section{Iterative Improvement: Use of the Newton-Raphson
Method}\label{sec:iteration}

This section and the next follow Roos et al~\cite{roos} (see also
Vishnoi~\cite{vishnoi}).

Let us assume that we have a solution $(x,y,s)$ to
$$Ax=b \mbox{ and } A^Ty+s=c \text{ and } x > 0 \text{ and } s > 0.$$
We will use the Newton-Raphson Method~\cite{roos} to get a ``better''
solution. Let us choose the next values as $x' = x+h$, $y' = y+k$, and
$s' = s+f$.  We can think of the steps $h$, $k$, and $f$ as small values.
Then we want, ignoring the positivity constraints for $x'$ and $s'$ for
the moment:
\begin{enumerate}
\item $Ax' = A(x+h)=b$, or equivalently, $Ax+Ah=b$. Since $Ax=b$, this
is tantamount to $Ah=0$.
\item $A^T y' + s' = A^T(y+k)+(s+f)=c$. Since $A^Ty+s=c$, we get 
$A^Tk+f=c-A^Ty-s=0$. 
\item  $x'_i s'_i = (x_i+h_i)(s_i+f_i)\approx \mu'$, or equivalently,
$x_is_i+h_is_i+f_ix_i+h_if_i\approx \mu'$. We drop the quadratic term
$h_i f_i$ (if the steps $h_i$ and $f_i$ are small, the quadratic term
$h_i f_i$ will be very small) and turn the approximate equality into an
equality, i.e., we require $x_is_i+h_is_i+f_ix_i=\mu'$ for all $i$. 
\end{enumerate}

\noindent Thus, we have a system of linear equations for $h_i,k_i,f_i$,
namely,
\begin{align*}
&&A h &= 0\\
\text{system (S)}&& A^T k + f &= 0\\
&&h_i s_i+f_i x_i &=\mu' - x_i s_i \quad \text{for all $i$}
\end{align*}
We show in Theorem~\ref{theorem} that system (S) can be solved by
``inverting'' a matrix. Note that there are $n$ variables $h_i$, $m$
variables $k_j$, and $n$ variables $f_i$ for a total of $2n + m$
unknowns. Also note that $Ah = 0$ constitutes $m$ equations, $A^T k + f =
0$ constitutes $n$ equations, and $h_i s_i+f_i x_i =\mu' - x_i s_i$ for
all $i$ comprises $n$ equations. So we have $2n + m$ equations and the
same number of unknowns. Also note that the $x_i$ and $s_i$ are
\emph{not} variables in this system, but fixed values. 

Before we show that the system has a unique solution, we make some simple
observations. From the third group of equations, we conclude \medskip

\begin{claim} \label{obs:lp:remark:1}
$(x_i+h_i)(s_i+f_i) = \mu'+h_if_i$, and $(x + h)^T (s + f) = n \mu' + h^T
f$. 
\end{claim}
\begin{proof} From the third group of equations, we obtain 
\[ (x_i+h_i)(s_i+f_i) = x_is_i+h_is_i+f_ix_i+h_if_i = \mu' + h_i f_i.\] 
Summation over $i$ yields \[ (x + h)^T (s + f) = \sum_i
(x_i+h_i)(s_i+f_i) = \sum_i \left(\mu' + h_i f_i\right) = n \mu' + h^T f.
\vspace{-5ex}\]
\end{proof}

\begin{claim} \label{claim:hifi:1}
$h^Tf=f^Th= \sum_i h_i f_i = 0$, i.e., the vectors $h$ and $f$ are
orthogonal to each other. 
\end{claim}
\begin{proof} Multiplying $A^Tk+f=0$ by $h^T$ from the left, we obtain
$h^TA^Tk+h^Tf=0$. Since $h^TA^T=(Ah)^T=0$, the equality $h^Tf=0$ follows.
\end{proof}

\begin{claim}\label{second}
$c^T (x + h) - b^T (y + k)= (x+h)^T(s+f)=n\mu'$.
\end{claim}
\begin{proof} From Claims~\ref{obs:lp:remark:1} and~\ref{claim:hifi:1},
$(x+h)^T(s+f)= n \mu' + h^T f = n \mu'$. Also, applying Claim~\ref{first}
to the primal solution $x' = x+h$ and to the dual solution $(y',s') =
(y+k,s+f)$ yields $c^T (x + h) - b^T (y + k) = (x + h)^T (s + f)$.
\end{proof}

Note that $n \mu'$ is the objective value gap of the updated solution.

\begin{theorem} \label{theorem} The system (S) has a unique solution. 
\end{theorem}
\begin{proof} We will follow Vanderbei~\cite{van} and use capital
letters (e.g. $X$) in this proof (only) to denote a diagonal matrix with
entries of the corresponding row vector (e.g. $X$ has the diagonal
entries $x_1,x_2,{\ldots} ,x_n$). We will also use $e$ to denote a column
vector of all ones (usually of length $n$).

Then, in the new notation, the last group of equations becomes
$$S h+Xf=\mu' e-XSe.$$
Let us look at this equation in more detail.
\begin{align*}
Sh+Xf &=\mu' e-XSe\\ 
h+ S^{-1}Xf &= S^{-1}\mu' e- S^{-1}XSe &\mbox{ pre-multiply by } S^{-1}\\ 
h+S^{-1}Xf &= \mu' S^{-1}e -X {S^{-1}}{S} e &\mbox{ diagonal
matrices commute}\\
h+S^{-1}Xf &= \mu' S^{-1}e -x &\mbox{ as } Xe=x\\
Ah+ AS^{-1}Xf &= \mu' AS^{-1}e-Ax  &\mbox{ pre-multiply by } A\\
AS^{-1}Xf &= \mu' AS^{-1}e-b &\mbox{ since } Ax=b \mbox{ and } Ah=0\\
-AS^{-1}XA^Tk &= \mu' AS^{-1}e-b &\mbox{ using } f=-A^Tk\\
b-\mu' AS^{-1}e &= (AS^{-1}XA^T)k
\end{align*}
As $XS^{-1}$ is diagonal with positive items, the matrix
$W=\sqrt{XS^{-1}}$ is well-defined. Note that the diagonal terms are
$\sqrt{x_i/s_i}$; since $x > 0$ and $s > 0$, we have $x_i/s_i > 0$ for
all $i$. Thus, $AS^{-1}XA^T = A W^2 A^T = (AW)(AW)^T$. Since $A$ has full
rank, $(AW)(AW)^T$, and hence $AS^{-1}XA^T$, is invertible (see
Appendix). Thus, 
\[ k=(AS^{-1}XA^T)^{-1}\left(b-\mu' AS^{-1}e\right).\] %
Then, we can find $f$ from $f=-A^Tk$. And to get $h$, we use the
equation: $h+S^{-1}Xf = \mu' S^{-1}e -x$, i.e., %
\[ h= -XS^{-1}f+\mu' S^{-1}e-x.\] %
Thus, system $(S)$ has a unique solution. \end{proof}

What have we achieved at this point? Given feasible solutions $(x,y,s)$
to the primal and dual problem, we can compute a solution $(x',y',s') =
(x + h, y + k, s + f)$ to $Ax' = b$ and $A^T y' + s' = c$ that also
satisfies $h^T f = 0$ and $ x'^T s = n \mu'$ for any prescribed parameter
$\mu'$. Why do we not simply choose $\mu' = 0$ and be done? It is because
we have ignored that we want $x' > 0$ and $s' > 0$. We will attend to
these constraints in the next section.

\section{Invariants in each Iteration}\label{sec:invariants}

Recall that we want to construct solutions $(x,y,s)$ to $P_\mu$ for
smaller and smaller values of $\mu$.  The solution to $P_\mu$ will
satisfy the following invariants. The first two invariants state that $x$
is a positive solution to the primal and $(y,s)$ is a solution to the
dual with positive $s$. The third invariant formalizes the condition $x_i
s_i \approx \mu$ for all $i$. 
\begin{description}
\item[(I1)] (primal feasibility) $Ax=b$ with $x>0$ (strict inequality).
\item[(I2)] (dual feasibility) $A^Ty+s=c$ with $s>0$ (strict inequality).
\item[(I3)] $\sigma^2 \define \sum_i \left(\frac{x_i
s_i}{\mu}-1\right)^2 \le \frac{1}{4}$. 
\end{description}

\begin{rem} Even though the variance of ${x_is_i}$ is $\frac{1}{n} \sum_i
\left(x_i s_i-\mu \right)^2$, we still use the notation
$\sigma^2$. \end{rem}

\noindent We need to show \[ x' > 0 \text{ and } s' > 0 \text{ and }
\sigma'^2 \define \sum_ i \left(\frac{x'_i s'_i}{\mu'} - 1 \right)^2 \le
\frac{1}{4}.\] %
We will do so for $\mu' = (1 - \delta) \mu$ and
$\delta=\Theta\left(\frac{1}{\sqrt{n}}\right)$.
Claim~\ref{obs:lp:remark:1} gives us an alternative expression for
$\sigma'^2$, namely,
\begin{equation}\label{sigma-prime-alternative}
\sigma'^2=\sum_i
\left(\frac{(x_i+h_i)( s_i+f_i)}{\mu'}-1\right)^2=\sum_i
\left(\frac{h_if_i}{\mu'}\right)^2
\end{equation}
We first show that the positivity invariants hold if $\sigma'$ is less
than one.

\begin{claim} \label{fact:int:1:1}
If $\sigma'<1$, then $x' >0$, and $s' >0$.
\end{claim}
\begin{proof} We first 
show that if $\sigma'<1$ then each product $x'_i s'_i =
(x_i+h_i)(s_i+f_i) = \mu' + h_i f_i$ is positive.  From $\sigma'<1$, we
get $\sigma'^2< 1$. Since $\sigma'^2 = \sum_i \left({h_if_i}/{\mu'}
\right)^2$, each term of the summation must be less than one, and hence,
$-\mu'< h_if_i<\mu'$. In particular, $\mu'+h_if_i>0$ for every $i$.
Thus, each product $(x_i+h)(s_i+f)$ is positive. 

Assume for the sake of a contradiction that both $x_i+h_i<0$ and
$s_i+f_i<0$. But as $s_i>0$ and $x_i>0$, this implies
$s_i(x_i+h_i)+x_i(s_i+f_i)<0$, or equivalently, $\mu' +x_is_i<0$, which
is impossible because $\mu',x_i,s_i$ are all non-negative. This is a
contradiction. \end{proof}

We next show $\sigma' \le 1/2$. We first establish
\begin{claim}\label{helper} $\frac{\mu}{x_is_i} \le \frac{1}{1 -
\sigma}$ for all $i$ and $\sum_i \abs{1 - \frac{x_i s_i}{\mu}} \le
\sqrt{n} \cdot \sigma$. \end{claim}
\begin{proof} As $\sigma^2 = \sum_i \left(1 - {x_i s_i}/{\mu}\right)^2$,
each individual term in the sum is at most $\sigma^2$. Thus, $\abs{1 -
{x_i s_i}/{\mu}} \le \sigma$, and hence, ${x_is_i}/{\mu} \ge 1 - \sigma$,
and further, ${\mu}/{x_is_i} \le {1}/(1 - \sigma)$.

For the second claim, we have to work harder. Consider any $n$ reals
$z_1$ to $z_n$. Then $(\sum_i \abs{z_i})^2 \le n \sum_i z_i^2$; this is
the frequently used inequality between the one-norm and the two-norm of a
vector\footnote{Indeed, \[ n \sum_i z_i^2 - \left(\sum_i z_i\right)^2 = n
\sum_i z_i^2 - \sum_i z_i^2 - 2\sum_{i < j} z_i z_j = (n - 1) \sum_i
z_i^2 - 2\sum_{i < j} z_i z_j = \sum_{i < j} (z_i - z_j)^2 \ge 0.\]}. We
apply the inequality with $z_i = 1 - {x_i s_i}/{\mu}$ and obtain the
second claim.
\end{proof}

\noindent Let us define two new quantities %
\[ H_i = h_i\sqrt{\frac{s_i}{x_i\mu'}} \quad\text{and}\quad F_i =
f_i\sqrt{\frac{x_i}{s_i\mu'}}. \]

\noindent Observe that $\sum_i H_iF_i=\sum \frac{h_if_i}{\mu'}=0$ (from
Claim~\ref{claim:hifi:1}) and $\sum_i (H_i F_i)^2 = \sum_i
\left(\frac{h_i f_i}{\mu'}\right)^2 = \sigma'^2$.  Also,

\begin{align}
H_i  + F_i &= \sqrt{ \frac{1}{x_i s_i \mu'}} \left(h_i s_i + f_i x_i\right)
= \sqrt{ \frac{1}{x_i s_i \mu'}} \left(\mu' - \mu + \mu - x_i s_i\right)
\nonumber\\
&= \sqrt{ \frac{\mu}{x_i s_i}\frac{\mu}{\mu'}} \left(\frac{\mu'}{\mu} -1
+ 1 - \frac{x_i s_i}{\mu}\right) = \sqrt{\frac{\mu}{x_is_i (1 - \delta)}}
\left(-\delta + 1 - \frac{x_i s_i}{\mu}\right). \label{Hi + Fi}
\end{align}
Finally,
\begin{align*}
\sigma'^2 &= \sum_i (H_i F_i)^2 = \frac{1}{4}  \left(\sum_i (H_i^2 + F_i^2)^2 - \sum_i (H_i^2 - F_i^2)^2\right)\\
&\le \frac{1}{4} \sum_i (H_i^2 + F_i^2)^2    &\text{since $\sum_i (H_i^2 - F_i^2)^2 \ge 0$}\\
&\le \frac{1}{4} \left(\sum_i (H_i^2 + F_i^2) \right)^2  &\text{more positive terms}\\ 
&= \frac{1}{4}  \left(\sum_i (H_i + F_i)^2  \right)^2   &\text{ since $H^T F = 0$}\\
&=\frac{1}{4} \left(  \sum_i    \frac{\mu}{x_is_i (1 - \delta)} \left(-\delta + 1 - \frac{x_i s_i}{\mu}\right)^2               \right)^2 &\text{by \lref{Hi + Fi}}\\
&\le \frac{1}{4 (1 - \delta)^2 (1 - \sigma)^2}  \left(  \sum_i  \left(-\delta + 1 - \frac{x_i s_i}{\mu}\right)^2  \right)^2 &\text{since $\mu/(x_is_i) \le 1/(1 - \sigma)$}\\
&\le \frac{1}{4 (1 - \delta)^2 (1 - \sigma)^2}  \left( n \delta^2 -
2\delta \sum_i \left(1 - \frac{x_is_i}{\mu}\right) + \sum_i \left(1 -\frac{x_i s_i}{\mu}\right)^2 \right)^2 &\text{remove inner square}\\
&\le \frac{1}{4 (1 - \delta)^2 (1 - \sigma)^2}  \left( n \delta^2 +
2\delta \sum_i \abs{1 - \frac{x_is_i}{\mu}} + 
\sum_i \left(1 - \frac{x_i s_i}{\mu}\right)^2 \right)^2\\ 
&\le \frac{1}{4 (1 - \delta)^2 (1 - \sigma)^2}  \left( n \delta^2 + 2\delta \sqrt{n} \cdot \sigma +  \sigma^2 \right)^2 &\text{by Claim~\ref{helper}}\\
&=  \frac{1}{4 (1 - \delta)^2 (1 - \sigma)^2}  \left( \left( \sqrt{n} \delta + \sigma\right)^2 \right)^2, &\text{forming inner square}
\end{align*}
and hence, 
\begin{equation}\label{defining inequality} 
\sigma' \le \frac{\left(\sqrt{n} \delta + \sigma\right)^2}{2 (1 -
\sigma)(1 - \delta)} \le \frac{\left(\sqrt{n} \delta + 1/2 \right)^2}{2
(1 - 1/2) (1 - \delta)} \stackrel{!}{\le} \frac{1}{2},\end{equation} 
where the second inequality holds since the bound for $\sigma'$ is
increasing in $\sigma$, and $\sigma \le 1/2$. We need to choose $\delta$
such that the last inequality holds. This is why we put an exclamation
mark on top of the $\le$-sign. Setting $\delta = {c}/\!\sqrt{n}$ for some
to be determined constant $c$ yields the requirement \[ \frac{\left( c +
1/2 \right)^2}{(1 - \delta)} \stackrel{!}{\le} \frac{1}{2}, \quad\text{or
equivalently,}\quad \left( 2c + 1 \right)^2 \stackrel{!}{\le} 2 \left(1 -
\frac{c}{\sqrt{n}}\right).\] This holds true for $c = 1/8$ and all $n \ge
1$. Thus, $\delta = 1/(8 \sqrt{n})$. 

\begin{rem} Why do we require $\sigma \le 1/2$ in the invariant? Let us
formulate the bound as $\sigma \le \sigma_0$ for some to be determined
$\sigma_0$. Then, the inequality~\lref{defining inequality} becomes \[
\frac{\left(\sqrt{n} \delta + \sigma_0 \right)^2}{2 (1 - \sigma_0) (1 -
\delta)} \stackrel{!}{\le} \sigma_0.\] We want this to hold for $\delta =
\frac{c}{\sqrt{n}}$ and some $c > 0$. In order for the inequality to hold
for $c = 0$, we need $\sigma_0 \le 2(1 - \sigma_0)$, or equivalently,
$\sigma_0 \le 2/3$. Since we want it to hold for some positive $c$, we
need to choose a smaller $\sigma_0$; $1/2$ is a nice number smaller than
$2/3$. \end{rem}

\paragraph{An Alternative Proof for Invariant (I3) (provided by Andreas
Karrenbauer)} Andreas Karrenbauer derived an alternative proof for
invariant (I3) that avoids introduction of the quantities $H$ and $F$ and
is more compact than the above. 

\begin{lemma}\label{alternative proof} 
Assume $\delta \le 1/6$. Then $\sigma \le \delta$ implies $\sigma' \le
\delta$.
\end{lemma}

\begin{proof} As $\sigma^2=\sum
\left(\frac{x_is_i}{\mu}-1\right)^2\le \delta^2$, each individual term
must be bounded by $\delta^2$. Thus, $\sigma \le \delta$ implies
$\left|\frac{x_is_i}{\mu}-1\right|\le \delta$, or $-\delta\leq
\frac{x_is_i}{\mu}-1$ or 
$x_i s_i \ge (1 - \delta) \mu$.

We define
$$\onenorm{\sigma'} = \sum_i \abs{\frac{(x_i + h_i)(s_i +f_i)}{\mu'} -
1}$$
Then from the definition of $\mu'$ and triangle inequality, 
\[ \onenorm{\sigma'} = \sum_i \abs{\frac{(x_i + h_i)(s_i +f_i)}{\mu'} -
1}\le \sum_i \abs{\frac{{x_i s_i + x_i f_i + h_i
s_i}}{{\mu'}} - 1} + \sum_i
\abs{\frac{ h_i f_i}{\mu'}} = \sum_i \abs{\frac{ h_i f_i}{\mu'}}.\] %
Again from $x_i f_i + s_i h_i = \mu' - x_i s_i$, we obtain (by squaring) 
\begin{equation}\label{hifi} 
h_i f_i = \frac{1}{2 x_i s_i} \left[ (\mu' - x_i s_i)^2 - (h_i s_i)^2 -
(x_i f_i)^2\right] \end{equation} Summing over $i$ and using the fact
that from Claim~\ref{claim:hifi:1}, $f^T h = 0$ we obtain
\begin{equation}\label{sum} 
\sum_i (\mu' - x_i s_i)^2 = \sum_i \left( (h_i s_i)^2 + (x_i f_i)^2
\right). \end{equation}

Assume that $\mu' = (1 - \tau) \mu$ for a $\tau$ to be fixed later. 
Then
\begin{align*}
\onenorm{\sigma'} 
&\le \sum_i \abs{\frac{ h_i f_i}{\mu'}} \\
&\le \sum_i \frac{1}{2 x_i s_i
\mu'} \left[ (\mu' - x_i s_i)^2 + (h_i s_i)^2 + (x_i f_i)^2\right] 
&\text{by~\lref{hifi}} \\
&= \sum_i \frac{ (\mu' - x_i s_i)^2}{x_i s_i \mu'} 
&\text{by~\lref{sum}}\\ 
&\le \sum_i \frac{\left(\frac{\mu'}{\mu} 
- \frac{x_i s_i}{\mu}\right)^2}{(1 - \delta) (1 - \tau)}
&\parbox{3cm}{\text{as } $x_i s_i \ge (1 - \delta) \mu$\\
\text{and } $\mu'=(1-\tau)\mu$}\\
&= \sum_i \frac{ \left(\tau - \left(\frac{x_is_i}{\mu}-1\right)\right)^2}{(1 
- \delta) (1 - \tau)}\\
&= \frac{n \tau^2 -2 \tau \sum \left(\frac{x_is_i}{\mu}-1\right)+ \sigma^2}{(1
-\delta)(1 - \tau)}\\
&\le 
\frac{n \tau^2 + 2 \tau \onenorm{\sigma} + \sigma^2}{(1 -\delta)(1 -
\tau)}\\
&\le \frac{(\sqrt{n} \tau + \delta)^2}{(1 - \delta)(1 - \tau)}
&\text{since $\onenorm{\sigma} \le \sqrt{n} \delta$ (Claim~\ref{helper})}\\
&\le \frac{4 \delta^2}{(1 - \delta)(1 - \delta/\sqrt{n})}         &\text{for the choice $\tau = \delta/\sqrt{n}$}\\
&\le \delta    &\text{for $\delta \le 1/6$}
\end{align*}
The claim follows as the two norm is always less than the one
norm\footnote{If $\alpha=(\alpha_1,{\ldots},\alpha_n)$ then $\left(\sum
|\alpha_i|\right)^2=\sum |\alpha_i|^2+2\sum_{i<j}|\alpha_i||\alpha_i|\geq
\sum |\alpha_i|^2=\sum \alpha_i^2$}, $\sigma'=\twonorm{\sigma'} \le
\onenorm{\sigma'}$. 
\end{proof}

\section{Initial Solution}
\label{sec:init}

This section follows Bertsimas and Tsitsiklis~\cite[p430]{BT}; see also
Karloff~\cite[p128-129]{karloff}. We have to deal with three problems:
\begin{enumerate}
\item how to make sure that we are dealing with a bounded problem
\item how to make sure that the problem is feasible and 
if the problem is feasible, then how to find an
initial solution 
\item how to guarantee condition (I3) for the initial
solution. 
\end{enumerate}

A standard solution for the second problem is the \emph{big M method}.
Let 
$x_0 \ge 0$ be an arbitrary nonnegative column vector of length $n$. 
We introduce a new variable $z \ge 0$, change $Ax = b$ into $Ax + (b-A
x_0) z = b$ and the objective into ``minimize $c^T x + Mz$'', where $M$
is a big number.  Note that 
$x = x_0$ 
and $z = 1$ is a feasible solution to the modified problem. We solve the
modified problem. If $z^*= 0$ in an optimal solution, we have also found
an optimal solution to the original problem. If $z^* > 0$ in an optimal
solution and $M$ was chosen big enough, the original problem is
infeasible. 

\begin{rem}
There are several other methods of dealing with the problem of getting a
starting solution. These include self-dual method~\cite{terlaky,van} and
the infeasible interior point method~\cite{MKD,ziang}.
\end{rem}

\emph{We assume for the remainder of the presentation that $A$, $b$, and
$c$ are integral and that $U$ is an integer with $U \ge \abs{a_{ij}},
\abs{b_i}, \abs{c_j}$ for all $i$ and $j$.}

\noindent We need the following Fact which we will prove in
Section~\ref{bounds}. 

\begin{fact}\label{bounded by W} 
Let $W = (mU)^m$.  If \lref{primal problem} is feasible, there is a
feasible solution with all coordinates bounded by $W$. If, in addition
the problem is bounded, there is an optimal solution with this property. 
\end{fact}

We now give the details. We add the constraint $e^T x+z \le (n+2) W$. If
the problem was feasible, it will stay feasible. If the problem was
bounded, the additional constraint does not change the optimal objective
value. If the problem was unbounded, the additional constraint makes it
bounded.  Using an additional slack variable $x_{n+1}$ we get the
equality $e^Tx+x_{n+1}+z= (n+2)W$. If we use ``normalized variables''
$x'_i=\frac{x_i}{W}$, drop the primes and use $x_{n+2}$ for $z$, we
obtain the following auxiliary primal problem. 
\begin{align}\label{artificial primal}
\begin{array}{r r c c c c c l  }
\text{minimize $c^Tx +Mx_{n+2}$, subject to}\quad & Ax & & & + & \rho
x_{n+2} & = & d\\
& e^T x &  + & x_{n+1} & + & x_{n+2} & = & n+2\\
& \multicolumn{2}{l}{x \ge 0} & \multicolumn{2}{l}{x_{n+1} \ge
0}&\multicolumn{2}{l}{x_{n+2} \ge 0},
\end{array}
\end{align}
where $d = \frac{1}{W} b$, $\rho = d - Ae$, and we show later in this
section, that $M$ can be chosen as $M = 4nU/\sep$, where $\sep =
\frac{1}{W^2} \cdot \frac{1}{2n \left((m+1) U\right)^{3(m+1)}}$. In
matrix form, the auxiliary primal is 
\[     A' \left( \begin{matrix} x \\ x_{n+1} \\ x_{n+2} \end{matrix}
\right) = b', \quad\text{where}\quad A' = \left(\begin{matrix}A & 0 & \rho \\
                                       e^T & 1 & 1 \end{matrix} \right)
				       \text{ and } b' = \left(
				       \begin{matrix} d \\ n+2
				       \end{matrix}\right) \] 
We make the following observations. 
\begin{compactenum}[\quad--]
\item As $x_i = 1$ for $1 \le i \le n+2$ is a feasible solution, 
\lref{artificial primal} is feasible. The feasible region is a polytope
contained in the cube defined by $0 \le x_i \le n+2$ for all $i$.  The
following Fact is shown in Section~\ref{integer}. 

\begin{fact}\label{nonzero coordinates} 
The nonzero coordinates of the vertices of this polytope are at least
$\sep$. \end{fact} \smallskip

\item As $0 \le x_i \le n+2$ and $c_i \ge -U$ for all $i$, the objective
value is at least $-U(n+2)$ Thus, \lref{artificial primal} is bounded.
\smallskip

\item If $x$ is feasible solution to \lref{primal problem} with 
$x_i \le W$ for $1\leq i\leq n$ then $(\frac{1}{W} x, (n+2) - \frac{1}{W}
e^T x, 0)$ is a feasible solution to~\lref{artificial primal} with
objective value $\frac{1}{W} c^T x$.

In particular, if \lref{primal problem} is feasible, then
\lref{artificial primal} has a solution with objective value less than or
equal to $nU$. This follows from $x_i/W \le 1$ and $c_i \le U$ for $1\le
i \le n$. 
\smallskip

\item We next show that if \lref{artificial primal} has an optimal
solution $(x^*, x^*_{n+1},x^*_{n+2})$ with $x^*_{n+2} = 0$ then
\lref{primal problem} is feasible. Indeed, $A W x^* = W A x^* = W d = b$
and hence $W x^*$ is feasible for \lref{primal problem}. If, in addition,
\lref{primal problem} is bounded, $W x^*$ is an optimal solution of
\lref{primal problem}. Note that if \lref{primal problem} is bounded, it
has an optimal solution $x$ with $x_i \le W$ by Fact~\ref{bounded by W}.
This solution induces a solution of~\lref{artificial primal} with
objective value $\frac{1}{W} c^T x$ by the preceding item. The optimality
of $(x^*, x^*_{n+1},x^*_{n+2})$ implies $c^T x^* \ge \frac{1}{W} c^T x$.
\smallskip

\item We finally show that if \lref{artificial primal} has an optimal
solution with $x^*_{n+2} > 0$, \lref{primal problem} is infeasible.
Indeed, then there must be an optimal vertex solution of \lref{artificial
primal}. For this vertex, $x^*_{n+2} \ge \sep$. The objective value of
this solution is at least $M\cdot \sep - (n+2) U = 2nU$. On the other
hand, if \lref{primal problem} is feasible, \lref{artificial primal} has
a solution with objective value at most $nU$. Any value of $M$ for which
$M\cdot \sep - (n+2) U > nU$ would work for this argument. $M = 4 U/R$ is
one such value. This explains the choice of $M$. 
\end{compactenum}
\smallskip

We summarize: Our original problem is feasible if and only if $x^*_{n+2}
= 0$ in every optimal solution to~\lref{artificial primal} if and only if
$x^*_{n+2} = 0$ in some optimal solution to~\lref{artificial primal}.
Moreover, if $x^*_{n+2} = 0$, and \lref{primal problem} is bounded,
$\frac{1}{W} x^*$ is an optimal solution of \lref{primal problem}.

\begin{rem} By the above, our original problem is feasible if and only
if $x^*_{n+2} = 0$ in an optimal solution to~\lref{artificial primal}. So
we can distinguish feasible and infeasible problems. How can we
distinguish bounded and unbounded problems? Note that the primal is
unbounded if it is feasible and the problem ``minimize 0 subject to $c^T
x = -1$, $Ax = 0$, and $x \ge 0$'' is feasible. So the test for
unboundedness reduces to two feasibility tests. 
 \end{rem}

\noindent The dual problem (with new dual variables $y_{m+1}, s_{n+1}$
and $s_{n+2}$) is
\begin{align}\label{artificial dual}
\text{maximize } & d^Ty+(n+2)y_{m+1}, \text{ subject to } & A^Ty+e
y_{m+1}+s &=c, \\ \nonumber&&\rho^T y+y_{m+1}+s_{n+2}& =M\\ \nonumber
&&y_{m+1}+s_{n+1}&=0 
\end{align}\vspace{-4ex}

\mbox{}\qquad\quad\ \ with slack variables $s\geq 0, s_{n+1} \ge 0,
s_{n+2} \ge 0$ and unconstrained variables $y$.\smallskip

Which initial solution should we choose? Recall that we also need to
satisfy (I3) for some choice of $\mu$, i.e., $\sum_{1 \le i \le n+2} (x_i
s_i/\mu - 1)^2 \le 1/4$. Also, recall that we set $x_i$ to $1$ for all
$i$. As $x_{n+1}=1$, we choose $s_{n+1}= {\mu}/{x_{n+1}}=\mu$. Then, from
the last equation, $y_{m+1}=-s_{n+1}=-\mu$. The simplest choice for $y$
is $y=0$. Then, from the first equation, $s=c+e\mu$, and from the second
equation $s_{n+2}=M-y_{m+1}=M+\mu$. Observe that all slack variables are
positive (provided $\mu$ is large enough). For this choice, 
\begin{align*}
\frac {x_i s_i}{\mu}-1&=\frac{c_i}{\mu}  &\text{for $i \le n$}\\
\frac{x_{n+1}s_{n+1}}{\mu} - 1 &=0\\
\frac{x_{n+2}s_{n+2}}{\mu}-1&=\frac{M}{\mu}.
\end{align*} 
Thus, $\sigma^2=\left(M^2+\sum c^2_i\right)/\mu^2$. We can make $\sigma^2
\le {1}/{4}$ by choosing 
\begin{equation}\label{definition of mu0}
\mu^2= 4\left(M^2+\sum c^2_i\right).\end{equation}

\paragraph{Summary:} Let us summarize what we have achieved. 
\begin{compactenum}[\hspace{\parindent}--]
\item For the auxiliary primal problem and its dual, we have
constructed solutions $(x^{(0)},y^{(0)},s^{(0)})$ that satisfy the
invariants for $\mu^{(0)} = 2\left(M^2+\sum c^2_i\right)^{1/2}$. 
\item From the initial solution, we can construct a sequence of
solutions $(x^{(t)},y^{(t)},s^{(t)})$ and corresponding $\mu^{(t)}$ such
that 
\begin{compactenum}[\hspace{\parindent}--]
\item $x^{(t)}$ is a solution to the auxiliary primal, 
\item $(y^{(t)},s^{(t)})$ is a solution to its dual, 
\item $\mu^{(t)} = (1 - \delta) \cdot \mu^{(t-1)} = (1 - \delta)^t \cdot
\mu^{(0)}$, and $\sum_j \left(x_j^{(t)}s_j^{(t)}/\mu^{(t)} - 1\right)^2
\le 1/4$.
\end{compactenum}
For $t \ge 1$, the difference between the primal and the dual objective
value is exactly $(n+2) \mu^{(t)}$ (Claim~\ref{second}). The gap
decreases by a factor $1 - \delta = 1 - 1/(8\sqrt{n+2})$ in each
iteration, and hence, can be made arbitrarily small.
\end{compactenum}
In the next section, we will exploit this fact and show how to extract
the optimal solution. Before doing so, we show the existence of an
optimal solution.

\begin{rem} {Existence of an Optimal Solution:} This paragraph requires
some knowledge of calculus, namely continuity and accumulation point. Our
sequence $(x^{(t)},y^{(t)},s^{(t)})$ has an accumulation point (this is
clear for the sequence of $x^{i}$ since the $x$-variables all lie between
$0$ and $n+2$ and we ask the reader to accept it for the others). Then
there is a converging subsequence. Let $(x^*,y^*,s^*)$ be its limit
point. Then $x^*$ and $(y^*,s^*)$ are feasible solutions of the
artificial primal and its dual respectively, and $x_i^* s_i^* = 0$ for
all $i$ by continuity. \end{rem}

\section{Extracting an Optimal Solution}
\label{part}\label{sec:extraction}

We will show how to round an approximate solution for the auxiliary
problems for a sufficiently small $\mu$ to an optimal solution. This
section is similar to~\cite[Theorem 5.3]{ye} and to the approach
in~\cite[Section 3.3]{roos}. See also~\cite{greenberg}. The auxiliary
problem has $m + 1$ constraints in $n+2$ variables. The auxiliary dual
problem has $n + 2$ constraints in $m + 1 + n + 2$ variables. We use $x$
to denote the variables of the auxiliary primal including $x_{n+1}$ and
$x_{n+2}$, and $y$ and $s$ for the variable vectors of the dual
(including the additional variables). Moreover, we use $A$ for the entire
constraint matrix and $b$ for the full right hand side. So $A$ is $(m+1)
\times (n+2)$, $b$ is a $(m+2)$-vector and $c$ is a $(n+2)$-vector. 

Consider an iterate $(x,y,s,\mu)$. We will first show that
$x_i \ge x^*_i/(4(n+2))$ and $s_i \ge s_i^*/(4(n+2))$ for all optimal
solutions $x^*$ and $(y^*,s^*)$ (Lemma~\ref{conclude zero}), i.e., if
$x_i^* > 0$ ($s_i^* > 0$) for some $i$, then $x_i$ ($s_i)$ cannot become
arbitrarily small. However, since $x_i s_i \le 2\mu$ always and $\mu$
decreases exponentially, at least one of $x_i$ or $s_i$ has to become
arbitrarily small. We use this observation to conclude that if $x_i$ is
sufficiently small (Lemma~\ref{B and N} quantifies what sufficiently
small means) then $x^*_i = 0$ in every optimal primal solution.
Similarly, if $s_i$ is sufficiently small, then $s^*_i = 0$ in every
optimal dual solution. 

Let $N$ be the set of indices for which we can conclude $x^*_i = 0$ and
let $B$ be the set of indices for which we can conclude $s^*_i = 0$.  We
show $B \cup N = \sset{1,\ldots,n}$ and $B \cap N = \emptyset$. We split
our last iterate $\bx$ into two parts $\bx_B$ and $\bx_N$ accordingly,
round the $N$-part to zero and recompute the $B$-part. Since the
coordinates in the $N$-part are tiny, this has little effect on the
$B$-part and hence the solution stays feasible. It stays optimal because
of complementary slackness.

\begin{lemma}\label{conclude zero}
Let $(x,y,s,\mu)$ satisfy (I1) to (I3).
\begin{enumerate}
\item For all $i \in \sset{1,\ldots,n}$: $x_i \ge x^*_i/(4(n+2))$ for
every optimal solution $x^*$ of the auxiliary primal.
\item For all $i \in \sset{1,\ldots,n}$: $s_i \ge s^*_i/(4(n+2))$ for
every optimal solution $(y^*,s^*)$ of the auxiliary dual.
\end{enumerate}
\end{lemma}
\begin{proof} By (I1) and (I2), $x$ is a feasible solution of the
auxiliary primal and $(y,s)$ a feasible solution of the auxiliary dual.
By (I3), we have $\sigma^2 = \sum_i (\frac{x_is_i}{\mu} - 1)^2 \le
\frac{1}{4}$. Thus, $(\frac{x_is_i}{\mu} - 1)^2 \le \frac{1}{4}$, and
hence, $\mu/2 \le x_i s_i \le 3\mu/2 < 2\mu$ for all $i$. Further, $x^T s
= \sum_i x_i s_i < 2(n+2)\mu$. 

Let $x^*$ be any optimal solution of the primal. Then $c^T x \ge c^T
x^*$. We apply Claim~\ref{weak duality} first to the solution pair $x$
and $(y,s)$ and then to the pair $x^*$ and $(y,s)$ to obtain %
\[ x^T s = c^T x - b^T y \ge c^T x^* - b^T y = (x^*)^T s. \]
Consider any $i \in \sset{1,\ldots,n+2}$ and assume $x_i <
x^*_i/(4(n+2))$. Since $x_i s_i \ge \mu/2$, we have $s_i \ge \mu/(2x_i) >
2 (n+2) \mu/x^*_i$, and hence 
\[ (x^*)^T s \ge x_i^* s_i > 2(n+2) \mu \ge x^T s \ge (x^*)^T s,\]
a contradiction.

Let $(y^*,s^*)$ be any optimal solution of the dual. Then $b^T y^* \ge
b^T y$. We apply Claim~\ref{weak duality} first to the solution pair $x$
and $(y,s)$ and then to the pair $x$ and $(y^*,s^*)$ to obtain %
\[x^T s = c^T x - b^T y \ge c^T x - b^T y^* = x^T s^*.\] Consider any $i
\in \sset{1,\ldots,n+2}$ and assume $s_i < s^*_i/(4(n+2))$.  Since $x_i
s_i \ge \mu/2$, we have $x_i \ge \mu/(2s_i) > 2 (n+2) \mu/s^*_i$, and
hence \[ x^T s^* \ge x_i s^*_i > 2(n+2)\mu \ge x^T s \ge x^T s^*,\] a
contradiction.
\end{proof} 

The preceding Lemma implies strong duality, one of the cornerstones of
linear programming theory. 

\begin{theorem}[Strong Duality] \label{strong:dual}
For each $i$, either $x^*_i=0$ in every optimal solution or $s^*_i=0$ in
every optimal solution. Thus, $c^T x^* - b^T y^* = (x^*)^T s^* = 0$. 
\end{theorem}

\begin{proof} Let $x^*$ and $(y^*,s^*)$ be any pair of optimal
solutions. Assume that there is an $i$ such that $x^*_i s^*_i > 0$. Let
$(x,y,s,\mu)$ satisfy the invariants (I1) to (I3). Then $x_i \ge
x^*_i/(4(n+2))$ and $s_i \ge s^*_i(4(n+2))$ by Lemma~\ref{conclude zero}.
Thus $2\mu > x_i s_i \ge x^*_i s^*_i /(16(n+2)^2)$. For $\mu < x^*_i
s^*_i /(32(n+2)^2)$, this is a contradiction.
\end{proof}

\begin{rem} We leave it to the reader to derive strong duality for the
original primal and dual from this. \end{rem}

By the Strict Complementarity Theorem (see e.g.~\cite[pp 77-78]{saigal}
or \cite[pp 20-21]{ye}), there are optimal solutions $x^*$ and
$(y^*,s^*)$ in which $x^*_i>0$ or $s^*_i>0$ for every $i$.  A
Quantitative version of strict complementarity is next stated in
Fact~\ref{strict complementarity} (the proof is in Section~7).

\begin{fact}\label{strict complementarity} Let $Q = R/(n+2)$. Then there
are optimal solutions $x^*$ and $(y^*,s^*)$ such that for all $i$ either
$x^*_i \ge Q$ and $s_i^* = 0$ or $s^*_i \ge Q$ and $x^*_i = 0$.
\end{fact}

\paragraph{The Rounding Procedure:} 

Throughout this section $x^*$ and $(y^*,s^*)$ denote optimal solutions as
in Fact~\ref{strict complementarity}.  We run the iterative improvement
algorithm until 
\begin{equation}\label{muf}
\mu < \mu_f \define R \cdot Q/(64 (n+2)^2 ((m+1) U)^{m+2}.
\end{equation}
Let $(\bx, \by, \bs, \bmu)$ be the last iterate. Let
\[B = \set{i}{\bs_i < Q/(4(n+2))} \quad\text{and}\quad
N = \set{i} {\bx_i < Q/(4(n+2))}.
\]

\begin{lemma}\label{B and N} 
$B \cup N = \sset{1,\ldots,n}$, $B \cap N = \emptyset$, $x^*_i = 0$ and
$\bx_i < 8 \bmu/ Q$ for every $i \in N$ and $s^*_i = 0$ and $\bs_i < 8
\bmu/Q$ for every $i \in B$. \end{lemma}

\begin{proof} Since $x_i s_i < 2\mu$ and $\mu \le Q^2/(32 n^2)$, we have
either $\bx_i < Q/(4(n+2))$ or $\bs_i < Q/(4(n+2))$. Thus $B \cup N =
\sset{1,\ldots,n}$. Since $\bx_i \ge x^*_i/(4(n+2))$ and $\bx_i \ge
s^*_i/(4(n+2))$ and either $x^* \ge Q$ or $s^*_i \ge Q$, we have $B \cap
N = \emptyset$. Consider any $i \in B$. Then $\bs_i < Q/(4(n+2))$ and
hence $s^*_i < Q$. Thus $s^*_i = 0$. Similarly, $i \in N$ implies $x^*_i
= 0$. Finally, since $\bx_i \bs_i < 2 \bmu$, we either have $\bx_i \ge
Q/(4(n+2))$ and $\bs_i < 8 \bmu/Q$ or $\bs_i \ge Q/(4(n+2))$ and $\bx_i <
8 \bmu/Q$. 
\end{proof} 

\newcommand{\rank}{\mathop{\mathrm{rank}}}

We split the variables $x$ into $x_B$ and $x_N$ and the matrix $A$ into
$A_B$ and $A_N$. Then our primal constraint system (ignoring the
non-negativity constraints) becomes
\[      A_B x_B + A_N x_N = b .\]
$(x^*_B, x^*_N)$ and $(\bx_B,\bx_N)$ are solutions of this system, and
$x^*_N = 0$ by Lemma~\ref{B and N}. Thus $A_B x^*_B = b$. 

Let us concentrate on the equation $A_B x_B = b$. If it has a unique
solution, call it $\tx_B$, then $\tx_B = x^*_B$.  We can find $\tx_B$ by
Gaussian elimination and $(\tx_B,0)$ will be the optimal solution and we
are done.

What can we do if $A_B x_B = b$ has an entire solution set? Then the rank
of the matrix $A_B$ is smaller than the cardinality of $B$. Let $B_1
\subseteq B$ be such that the rank of the matrix $A_{B_1}$ is equal to
the cardinality of $B_1$ and let $B_2 = B \setminus B_1$. We can find
$B_1$ by Gaussian elimination. Then our system becomes 
\[     A_{B_1} x_{B_1} + A_{B_2} x_{B_2} + A_N x_N = b.\]
For every choice of $x_{B_2}$ and $x_N$ this system has a unique
solution\footnote{Let $m' \le m$ be the rank of $A_B$. By row operations
and permutation of columns, we can transform the system $A_B x_B + A_N
x_N = b$ into %
\[ \begin{array}{ccccccc} I x_{B_1} &+ & A'_{B_2} x_{B_2} & + & A'_N x_N
& = & b' \\ 0 &+& 0 &+& A''_N x_N & = & b'', 
\end{array} \] %
where $I$ is a $m' \times m'$ identity matrix, $A'_{B_2}$, $A'_N$, and
$b'$ have $m'$ rows, and $A''_N$ and $b''$ have $m - m'$ rows. Since
$(x^*_B, x^*_N)$ is a solution to this system and $x^*_N = 0$, we have
$b'' = 0$. Since $(\bx_B,\bx_N)$ is a solution to this system, we have 
further $A''_N \bx_N = 0$. Thus for every choice of $x_{B_2}$ and $x_N$
this system has a unique solution for $x_{B_1}$. } %
for $x_{B_1}$. Let $\tx_{B_1}$ be the solution of 

\[ A_{B_1} \tx_{B_1} + A_{B_2} \bx_{B_2} = b \quad (\text{$x_N$ is set to
zero and $x_{B_2}$ is set to $\bx_{B_2}$}).\]
Subtracting this equation from $A_{B_1} \bx_{B_1} + A_{B_2} \bx_{B_2} +
A_N \bx_N = b$ yields
\[   A_{B_1} (\bx_{B_1} - \tx_{B_1})  + A_N \bx_N = 0. \]
The coordinates of $\bx_N$ are bounded by $8 \bmu/Q$ and hence the
coordinates of $A_N \bx_N$ are bounded by $8 (n+2) U \bmu/ Q = R/(8 (n+2)
((m + 1 ) U)^{m + 1})$ in absolute value. By the remark after
Lemma~\ref{solutions to a linear system} of Section~7, all coordinates of
$\bx_{B_1} - \tx_{B_1}$ are bounded by $((m+1) U)^{m+1}$ times this
number in absolute value, i.e., are bounded by $R/(8(n+2))$ in absolute
value. Since $\bx_i \ge \sep/(4(n+2))$ for every $i \in N$, we have
$\tx_{B_1} \ge 0$.  Thus $\tilde{x} = (\tx_{B_1}, \bx_{B_2}, 0)$ is a
feasible solution of~\ref{artificial primal}. Since $\tilde{x}^T s^* =
\sum_{i \in B} \tilde{x}_i s^*_i + \sum_{i \in N} \tilde{x}_i s^*_i = 0 +
0 = 0$, $\tilde{x}$ is an optimal solution to~\ref{artificial primal}.

\section{Complexity}

Let us assume that the initial value of $\mu$ is $\mu_0$ and that we want
to decrease $\mu$ to $\mu_f$.  Since every iteration decreases $\mu$ by
the factor $(1 - \delta)$, we have $\mu = (1-\delta)^r\mu_0$ after $r$
iterations. The smallest $r$ such that $(1 - \delta)^r \le \mu_f$ is
given by 
\[ \ln \frac{\mu_0}{\mu_f} = -r \ln(1-\delta)\approx -r
(-\delta),\]
or equivalently,
\[ r=O\left(\frac{1}{\delta}\log \frac{\mu_0}{\mu_f}\right)=
O\left(\sqrt{n}\log \frac{\mu_0}{\mu_f}\right).\]
In \lref{definition of mu0}, we defined
\[ \mu_0^2 = 4 \left(M^2+\sum c_i^2\right)\leq
4 \left( \frac{16n^2 U^2 }{\sep^2} + n U^2 \right) \le 68 \frac{n^2
U^2}{\sep^2}.\]
In \lref{muf}, we defined $\mu_f$. Thus, the number of
iterations will be
\begin{align*}
r &=O\left(\sqrt{n}\log \frac{\mu_0}{\mu_f}\right)=
O\left(\sqrt{n}\log \frac{n^2U^2/\sep^2}{RQ/(64 (n+2)^2 ((m+1)
U)^{m+2})}\right)\\
&=O\left(\sqrt{n}(\log n+ m\log (mU) + \log \frac{1}{\sep} 
\right) \\
&= O\left(\sqrt{n}(\log n+m \left(\log (mU)\right) \right), \end{align*}
as $\log \frac{1}{\sep} = O(\log n + m (\log (mU)))$.

\section{The Proofs of Facts 1 to 3}
\label{integer}\label{bounds}

In the previous sections, we used upper bounds on the components of an
optimal solution and lower bounds on the nonzero components of an optimal
solution. In this section, we derive these bounds.  In this section, we 
assumes more knowledge of linear algebra, namely, determinants and
Cramer's rule, and some knowledge of geometry. Unless stated otherwise,
we assume that all entries of $A$ and $b$ are integers bounded by $U$ in
absolute value.

\newcommand{\sign}{{\mathit sign}}

The determinant of a $k \times k$ matrix $G$ is a sum of $k!$ terms, 
namely, 
\[ \det G =\sum_\pi \sign(\pi) \cdot g_{1\pi(1)}g_{2\pi(2)}{\ldots}
g_{k\pi(k)}. \]
The summation is over all permutations $\pi$ of $k$ elements, $\sign(\pi)
\in \sset{-1,1}$, and the product corresponding to a permutation $\pi$
selects the $\pi(i)$-th element in row $i$ for each $i$.  Each product is
at most $U^k$. As there are $k!$ summands, we have $\abs{\det G} \le k!
U^k \le (kU)^k$; see~\cite[pp 373-374]{BT}, \cite[p75]{karloff} or
\cite[pp 43-44]{saigal}.

Cramer's rule states that the solution of the equation $Gz = g$ (for a
$k\times k$ non-singular matrix $G$) is $z_i=(\det G_i)/\det G$, where
$G_i$ is obtained by replacing the $i$th column of $G$ by $g$. 

\begin{lemma}\label{solutions to a linear system} 
Let $G z = g$ be a linear system in $k$ variables with a unique solution.
Let $z^*$ be the solution of the system. If all entries of $G$ and $g$
are integers bounded by $U$ in absolute value then $\abs{z^*_i} \le (kU)^k$ for
all $i$ and $z^*_i \not= 0$ implies $\abs{z^*_i} \ge
1/(kU)^k$.\end{lemma}

\begin{proof} Since the system has a unique solution there is a
subsystem $G' z = g'$ consisting of $k$ equations such that $G'$ is
non-singular and $G' z^* = g'$. Then $z^*_i = (\det G'_i)/\det G'$, where
$G'_i$ is obtained from $G'$ by replacing the $i$th column of $G$ by
$g'$. Since all entries of $G$ and $g$ are integral, $\det G'$ is at
least one in absolute value, $\det G'_i$ is at least one in absolute
value if nonzero, and $\det G'_i \le (kU)^k$. The bounds follow.
\end{proof}

If the entries of the right-hand side $g$ are bounded by $U'$ instead of $U$, the upper bound
becomes $k^k U^{k-1} U'$.

\begin{lemma}\label{F and O} Assume that \lref{primal problem} is
feasible. Let $x$ be a feasible solution with the maximum number of zero
coordinates (equivalently the minimum number of nonzero
coordinates).\footnote{Consider minimize 0 subject to $x_1 + x_2 = 1$,
$x_1 \ge 0$ and $x_2 \ge 0$. The feasible solutions $(0,1)$ and $(1,0)$
have one nonzero coordinate. The feasible solutions $(x_1,x_2)$ with $x_1
> 0$ and $x_2 > 0$ and $x_1 + x_2 = 1$ have two nonzero coordinates.} Let
$B$ be the set of indices for which $x_i \not= 0$, and let $A_B$ be the
submatrix of $A$ formed by the columns indexed by $B$. Then $A_B z = b$
has a unique solution, where the dimension of $z$ is equal to the number
of columns of $A_B$. 

If, in addition \lref{primal problem} is bounded, the same claim holds
for an optimal solution with a maximum number of zero coordinates.
\end{lemma}

\begin{proof} Let $x_B$ be the restriction of $x$ to the indices in $B$.
Then $A_B x_B = b$. Assume there is a second solution $x'_B$ of $A_B z =
b$ with $x'_B \not= x_B$. Then all points $z(\lambda) = x_B +
\lambda(x'_B - x_B)$, $\lambda \in \R$, satisfy $A_B z = b$. These points
form a line. Consider the intersection $z^*$ closest to $x_B$ of this
line with one of the coordinate planes $z_i = 0$; if there are several
with the same distance choose one of them. Then $z^* \ge 0$ because we
consider an intersection closest to $x_B$ and $z^*_i = 0$ for at least
one $i \in B$. Thus $z^*$ is a feasible solution to \lref{primal problem}
with one more zero coordinate, a contradiction to the definition of $x$. 

If \lref{primal problem} is bounded, there is an optimal solution. Let
$x$ be an optimum solution with a maximum number of zero coordinates.
Define $x_B$, $x_B'$, and $z(\lambda)$ as above. Since $x_B > 0$, the $z(\lambda)$ is
feasible for small enough $\abs{\lambda}$. Also $c_B^T z(\lambda) = c_B^T x_B
+ \lambda (c_B^T x'_B - c_B^T x_B)$; here $c_B$ is the restriction of $c$ to the indices in $B$. 
Since $\lambda$ may be positive or
negative, we must have $c_B^T x'_B = c_B^T x_B$ and hence $z(\lambda)$ is
feasible and optimal as long as $z(\lambda) \ge 0$. The proof is now
completed as in the preceding paragraph.\end{proof}

We can now give the proof of Facts~\ref{bounded by W}, \ref{nonzero
coordinates}, and~\ref{strict complementarity}. 

\begin{proof} (Fact~\ref{bounded by W}) Consider a feasible (optimal)
solution $x$ of~\lref{primal problem} with a maximum number of zero
coordinates. Then $x$ is of the form $x = (x_B, x_N)$ with $x_N = 0$ and
$x_B$ being the unique solution to the system $A_B x_B = b$. Thus the
coordinates of $x_B$ are bounded by $(m U)^m$.
\end{proof}

\begin{proof} (Fact~\ref{nonzero coordinates}) Let $x^*$ be an optimal
vertex of the artificial primal~\lref{artificial primal}. How small can a
nonzero coordinate of $x^*$ be? The constraint system is 
\[ \begin{array}{r c c c c c l } Ax & & & +
& (\frac{1}{W} b - Ae) x_{n+2} & = & \frac{1}{W} b\\ e^T x & + & x_{n+1}
& + & x_{n+2} & = & (n+2).
\end{array}\]
Let $B$ be the index set of the nonzero coordinates of $x^*$. Then
$x^*_B$ is the solution to a subsystem formed by $\abs{B}$ columns of the
above and this subsystem has a unique solution. For $i \in B$, $x^*_i =
\det G_i/\det G$, where $G$ is a nonsingular square matrix and $G_i$ is
obtained from $G$ by replacing the $i$th column by the corresponding
entries of the right hand side. In the system above, the entries in the
column corresponding to $x_{n+2}$ are bounded by $(n+1) U$, and all other
entries are bounded by $U$. Since any product in the determinant formula
for $G$ can contain only one value of the column for $x_{n+2}$, we have
$\abs{\det G} \le (m+1)! (n+1)U^{m+1}$. Consider next $\det G_i$. We need
to lower bound $\abs{\det G_i}$. The matrix $G_i$ may contain two columns
with fractional values. If we multiply these columns with $W$, we obtain
an integer matrix. Thus, $\abs{\det G_i} \ge 1/W^2$ if nonzero. Thus
\begin{equation}\label{definition of sep} 
x_i^* \ge \frac{1}{W^2} \cdot \frac{1}{2n \left((m+1) U\right)^{m+1}} \ge
\frac{1}{2 n \left((m+1) U\right)^{3(m+1)}}. 
\end{equation} \end{proof}

\newcommand{\cO}{{\cal O}}

\begin{proof} (Fact~\ref{strict complementarity}) We prove the fact for
the auxiliary primal. Let $\cO$ be a smallest set of optimal vertices
with the property that if for some $i$ there is an optimal solution with
$x^*_i > 0$, then $\cO$ contains an optimal vertex with this property.
Then $\abs{\cO} \le n + 2$. Let $x^{**} = \frac{1}{\abs{\cO}} \sum_{x^*
\in \cO} x^*$ be the center of gravity of the vertices in $\cO$. Then
$x^{**}_i \ge x^*_i/(n+2)$ for every $x^* \in \cO$. Thus $Q = R/(n+2)$
works. 
\end{proof}

\paragraph{Beyond the Integral Case} If the entries of $A$ and $b$ are
rational numbers, we write the entries in each column (or row) with a
common denominator. Pulling them out brings us back to the integral case.
For example,
\[   \left| \begin{array}{cc} {2}/{3}   & {4}/{5} \\
                                            {1}/{3}   & {6}/{5} \end{array} \right| = \frac{1}{15} \left| \begin{array}{cc} {2}  & {4}\\
                                            {1}  & {6} \end{array}\right|. \]
Thus, if the determinant is nonzero, it is at least $1/15$. 


\subsection*{Acknowledgments}

The first author thanks Andreas Karrenbauer and Ruben Becker for
intensive discussions and Andreas Karrenbauer for teaching from an
earlier draft of this note. The work of the second author was inspired
by an informal lecture given by Nisheeth Vishnoi at IIT Kanpur. The
second author also thanks the students of CS602 (2014-15 and 2015-2016
batches) for their helpful comments and questions.  Thanks also to Romesh
Saigal for very prompt replies to queries. 

The authors also thank Andreas Karrenbauer for providing the proof of
Lemma~\ref{alternative proof}.

\bibliographystyle{acm}

\subsection*{Appendix: Result from Algebra} 

Assume that $A$ is $m\times n$ matrix and the rank of $A$ is $m$, with
$m<n$. Then, all $m$ rows of $A$ are linearly independent. Or,
$\alpha_1A_1+\alpha_2A_2+{\ldots} +\alpha_mA_m=0$ ($0$ here being a row
vector of size $n$) has only one solution $\alpha_i=0$. Thus, if $x$ is
any $m \times 1$ matrix (a column vector of size $m$), then $x^TA=0$
implies $x=0$. Note that $(x^T A)^T = A^T x$. Thus, $A^T x = 0$ implies
$x = 0$. 

As $A$ is $m\times n$ matrix, $A^T$ will be $n\times m$ matrix. The
product $AA^T$ will be an $m\times m$ square matrix.

Consider the equation $(AA^T)x=0$. Pre-multiplying by $x^T$ we get
$x^TAA^Tx =0$ or $(A^T x)^T (A^T x) =0$. Now, $(A^T x)^T (A^T x)$ is the
squared length of the vector $A^T x$. If a vector has length zero, all
its coordinates must be zero. Thus, $A^T x = 0$, and hence, $x = 0$ by
the preceding paragraph. 

Thus, the matrix $AA^T$ has rank $m$ and is invertible.

Also observe that if $X$ is a diagonal matrix (with all diagonal entries
non-zero) and if $A$ has full row-rank, then $AX$ will also have full
row-rank. Basically, if the entries of $X$ are $x_1,x_2,{\ldots} ,x_m$
then the matrix $AX$ will have rows as $x_1A_1,x_2A_2,{\ldots} ,x_mA_m$
(i.e., $i$th row of $A$ gets scaled by $x_i$). If rows of $AX$ are not
independent, then there are $\beta$s (not all zero) such that 
$\beta_1x_1A_1+\beta_2x_2A_2+{\ldots} +\beta_mx_mA_m=0$, or there are
$\alpha$s (not all zero) such that $\alpha_1A_1+\alpha_2A_2+{\ldots}
+\alpha_mA_m=0$ with $\alpha_i=\beta_ix_i$.

\end{document}